\documentclass[12pt,oneside,english]{amsart}
\usepackage[T1]{fontenc}
\usepackage[utf8]{inputenc}
\usepackage[letterpaper]{geometry}
\usepackage{booktabs}
\usepackage{amstext}
\usepackage{amsthm}
\usepackage{amssymb}
\usepackage{graphicx}
\usepackage{setspace}
\usepackage[authoryear]{natbib}
\usepackage{subscript}
\makeatletter

\providecommand{\tabularnewline}{\\}

\numberwithin{equation}{section}
\numberwithin{figure}{section}
\theoremstyle{plain}
\newtheorem{thm}{\protect\theoremname}
  \theoremstyle{remark}
  \newtheorem*{rem*}{\protect\remarkname}

\theoremstyle{plain}

\numberwithin{equation}{section}
\usepackage[within=none]{caption}

\usepackage{subfigure}
\usepackage{babel}
\usepackage{amsmath}
\usepackage{algpseudocode}

\usepackage{lineno}
\usepackage{url}
\newcommand{\mockalph}[1]{}

\makeatother

\usepackage{babel}
  \providecommand{\remarkname}{Remark}
\providecommand{\theoremname}{Theorem}

\begin{document}

\title[Delayed acceptance MCMC by data subsampling]{Speeding up MCMC by delayed acceptance and data subsampling}

\author{Matias Quiroz\textsuperscript{$\star\dagger$}, Minh-Ngoc Tran\textsuperscript{$\ddagger$},
Mattias Villani\textsuperscript{$\star$} and Robert Kohn\textsuperscript{$\dagger\dagger$}}

\thanks{\textit{\textsuperscript{$^{\star}$}}\emph{Division of Statistics
and Machine Learning, Link\"{o}ping University.}\\
 $^{\dagger}$\textit{Research Division, Sveriges Riksbank.} \\
\textit{ }\textsuperscript{$\ddagger$}\emph{Discipline of Business
Analytics, University of Sydney.}\\
\textsuperscript{$\dagger\dagger$}\emph{Australian School of Business,
University of New South Wales}. \\
Quiroz was partially supported by VINNOVA grant 2010-02635. Tran was
partially supported by a Business School Pilot Research grant. Villani
was partially supported by Swedish Foundation for Strategic Research
(Smart Systems: RIT 15-0097). Kohn was partially supported by Australian
Research Council Centre of Excellence grant CE140100049. The views
expressed in this paper are solely the responsibility of the authors
and should not be interpreted as reflecting the views of the Executive
Board of Sveriges Riksbank. The authors would like to thank the Editor,
the Associate Editor and the reviewers for their comments that helped
to improve the manuscript.}
\begin{abstract}
The complexity of the Metropolis-Hastings (MH) algorithm arises from
the requirement of a likelihood evaluation for the full data set in
each iteration. \citet{payne2015bayesian} propose to speed up the
algorithm by a delayed acceptance approach where the acceptance decision
proceeds in two stages. In the first stage, an estimate of the likelihood
based on a random subsample determines if it is likely that the draw
will be accepted and, if so, the second stage uses the full data likelihood
to decide upon final acceptance. Evaluating the full data likelihood
is thus avoided for draws that are unlikely to be accepted. We propose
a more precise likelihood estimator which incorporates auxiliary information
about the full data likelihood while only operating on a sparse set
of the data. We prove that the resulting delayed acceptance MH is
more efficient compared to that of \citet{payne2015bayesian}. The
caveat of this approach is that the full data set needs to be evaluated
in the second stage. We therefore propose to substitute this evaluation
by an estimate and construct a state-dependent approximation thereof
to use in the first stage. This results in an algorithm that (i) can
use a smaller subsample $m$ by leveraging on recent advances in Pseudo-Marginal
MH (PMMH) and (ii) is provably within $O(m^{-2})$ of the true posterior.

\noindent \textsc{Keywords}: Bayesian inference, Markov chain Monte
Carlo, Delayed acceptance MCMC, Large data, Survey sampling

\noindent \newpage{}
\end{abstract}

\maketitle

\section{Introduction\label{sec:Introduction}}

Markov Chain Monte Carlo (MCMC) methods have been the workhorse for
sampling from nonstandard posterior distributions in Bayesian statistics
for nearly three decades. Recently, with increasingly more complex
models and/or larger data sets, there has been a surge of interest
in improving the $O(n)$ complexity emerging from the necessity of
a complete data scan in each iteration of the algorithm.

There are a number of approaches proposed in the literature to speed
up MCMC. Some authors divide the data into different partitions and
carry out MCMC for the partitions in a parallel and distributed manner.
The draws from each partition's subposterior are subsequently combined
to obtain an approximation of the full posterior distribution. This
line of work includes \citet{scott2013bayes,neiswanger2013asymptotically,wang2013parallel,minsker2014scalable,nemeth2016merging},
among others. Other authors use a subsample of the data in each MCMC
iteration to speed up the algorithm, see e.g. \citet{korattikara2014austerity},
\citet{bardenet2014towards}, \citet{maclaurin2014firefly}, \citet{maire2015light},
\citet{bardenet2015markov} and Quiroz et al. (\citeyear{quiroz2016speeding},
\citeyear{quiroz2016exact}). Finally, delayed acceptance MCMC has
been used to speed up computations \citep{banterle2014accelerating,payne2015bayesian}.
The main idea in delayed acceptance is to avoid computations if there
is an indication that the proposed draw will ultimately be rejected.
\citet{payne2015bayesian} consider a two stage delayed acceptance
MCMC that uses a random subsample of the data in the first stage to
estimate the likelihood function at the proposed draw. If this estimate
suggests that the proposed draw is likely to be rejected, the algorithm
does not proceed to the second stage that evaluates the true likelihood
(using all data). \citet{banterle2014accelerating} divide the acceptance
ratio in the standard Metropolis-Hastings (MH) \citep{metropolis1953equation,hastings1970monte}
into several stages, which are sequentially performed until a first
rejection is detected implying a final rejection of the proposed draw.
Clearly, both delayed acceptance approaches save computations for
proposals that are unlikely to be accepted, but on the other hand
require all stages to be computed for those likely to be accepted.
The latter corresponds to at least the same computational cost for
one iteration as the standard MH.

This paper extends the delayed acceptance algorithms of \citet{payne2015bayesian}
in the following directions. First, we replace the likelihood estimator
of the first stage by an efficient estimator that employs control
variates to significantly reduce the variance. We show that this modification
results in an algorithm which is more effective in promoting good
proposals to the second stage. Since this algorithm computes the true
likelihood (as does MH) in the second stage whenever a draw is up
for a final accept/reject decision, we refer to it as Delayed Acceptance
(standard) MH (DA-MH). Second, we propose a delayed acceptance algorithm
that overcomes the caveat of a full data likelihood evaluation in
the second stage by replacing it with an estimate. This second stage
estimate is initially developed in the approximate subsampling Pseudo-Marginal
MH (PMMH) framework in \citet{quiroz2016speeding}. Thus, our second
contribution is to speed up their algorithm by combining with delayed
acceptance and we document large speedups for our so called Delayed
Acceptance PMMH (DA-PMMH) algorithm compared to MH. \citet{payne2015bayesian}
instead propose to circumvent the full data evaluation by combining
their delayed acceptance sampler with the consensus Monte Carlo in
\citet{scott2013bayes}, which currently lacks any control of the
error produced in the approximate posterior. In contrast, we can make
use of results in \citet{quiroz2016speeding} to control the approximation
error and also ensure that it is $O(m^{-2})$, where $m$ is the subsample
size used for estimating the likelihood.

While the delayed acceptance MH has the disadvantage of using all
data whenever a proposed draw proceeds to the second stage, we believe
that the successful implementation provided here is essential if exact
simulation is of importance. By exact simulation we mean that the
Markov chain produced by the subsampling algorithm has the same invariant
distribution as that of a Markov chain produced by an algorithm that
uses all the data. Exact simulation using subsets of the data has
proven to be extremely challenging. Pseudo-marginal MCMC \citep{beaumont2003estimation,andrieu2009pseudo}
provides a framework for conducting Markov chain simulation by only
using an estimate (in this setting estimated from a subsample of the
data) of the likelihood. Remarkably, although the true likelihood
is never evaluated, the simulation is exact provided that the likelihood
estimate is unbiased and almost surely positive. One route to exact
simulation by data subsampling suggested by \citet{bardenet2015markov}
is to compute a sequence of unbiased estimators of the log-likelihood,
apply the technique in \citet{rhee2015unbiased} to debias the resulting
likelihood estimator and subsequently use it within a pseudo-marginal
MCMC. \citet{bardenet2015markov} note that, as proved by \citet{jacob2015nonnegative},
a lower bound on the log-likelihood estimators is needed to ensure
positiveness. This typically requires computations using the full
data set, naturally defeating the purpose of subsampling. \citet{quiroz2016exact}
instead suggests to compute a soft lower bound which is defined to
be a lower bound with a high probability. Since the estimate can occasionally
be negative, the pseudo-marginal in \citet{quiroz2016exact} instead
targets an absolute measure following \citet{lyne2015russian}, who
show that the draws can be corrected with an importance sampling type
of step to estimate expectations of posterior functions exactly. We
note that although this certainly is exact inference of the expectation
(it is not biased) the algorithm does not provide exact simulation
as defined above. \citet{maclaurin2014firefly} propose Firefly Monte
Carlo, which introduces an auxiliary variable for each observation,
which determines if it is included when evaluating the likelihood.
The distribution of these binary variables is such that when they
are integrated out, the marginal posterior is the same as the one
targeted by MH, thus providing exact simulation. Typically a small
fraction of observations are included, hence speeding up the execution
time significantly compared to MH. However, the augmentation scheme
severely affects the mixing of the Firefly algorithm and it has been
demonstrated to perform poorly compared to MH and other subsampling
approaches \citep{bardenet2015markov,quiroz2016speeding,quiroz2016exact}.
We conclude that delayed acceptance, out of the discussed methods,
seems to be the only feasible route to obtain exact simulation via
data subsampling. We demonstrate that our implementation is crucial
to obtain an algorithm that can improve on MH.

This paper is organized as follows. Section \ref{sec:Likelihood-estimators}
presents the efficient likelihood estimator. Sections \ref{sec:delayedMH}
and \ref{sec:delayedPMMH} outline the delayed acceptance MH and PMMH
methodologies in the context of subsampling. Section \ref{sec:Application}
applies the method to a micro-economic data set containing nearly
5 million observations. Section \ref{sec:Conclusions-and-Future}
concludes and Appendix \ref{sec:Comparing-the-Efficiency} proves
Theorem \ref{thm:theorem1}.

\section{Likelihood estimators with subsets of data\label{sec:Likelihood-estimators}}

\subsection{Structure of the likelihood\label{subsec:Structure-likelihood}}

Let $\theta$ be the parameter in a model with density $p(y_{k}|\theta,x_{k})$,
where $y_{k}$ is a potentially multivariate response vector and $x_{k}$
is a vector of covariates for the $k$th observation. Let $l_{k}(\theta)=\log p(y_{k}|\theta,x_{k})$
denote the $k$th observation's log-density, $k=1,\dots,n$. Given
conditionally independent observations, the likelihood function can
be written
\begin{equation}
p(y\vert\theta)=\exp\left[l(\theta)\right],\label{eq:likelihood}
\end{equation}
where $l(\theta)=\sum_{k=\text{1}}^{n}l_{k}(\theta)$ is the log-likelihood
function. To estimate $p(y\vert\theta)$, we first estimate $l(\theta)$
based on a sample of size $m$ from the population $\left\{ l_{1}(\theta),\dots,l_{n}(\theta)\right\} $
and subsequently use \eqref{eq:likelihood}. The first step corresponds
to the classical survey sampling problem of estimating a population
total: see \citet{sarndal2003model} for an introduction to survey
sampling. Note that \eqref{eq:likelihood} is more general than identically
independently distributed (iid) observations, although we require
that the log-likelihood can be written as a sum of terms, where each
term depends on a unique piece of data information. One example is
models with a random effect for each subject, with possibly multiple
individual observations per subject (longitudinal data). In this case,
a single term in the log-likelihood sum corresponds to the log joint
density for a subject, and we therefore sample subjects (rather than
individual observations) when estimating \eqref{eq:likelihood}.

\subsection{Data subsampling\label{subsec:Data-subsampling}}

A main distinction of sampling schemes is whether the sample is obtained
with or without replacement. It is clear that sampling with replacement
gives a higher variance for any function of the sample: including
the same element more than once does not provide any further information
about finite population characteristics. However, it results in sample
elements that are independent which facilitates the derivation of
Theorem \ref{thm:theorem1} for the delayed acceptance MH. It also
allows us to apply the theory and methodology in \citet{quiroz2016speeding}
for the delayed acceptance PMMH in Section \ref{sec:delayedPMMH}.
It should be noted that the two sampling schemes are approximately
the same when $m\ll n$.

Let $u=(u_{1},\dots u_{m})$ be a vector of indices obtained by sampling
$m$ indices with replacement from $\{1,\dots,n\}$ and let $\left\{ l_{u_{1}}(\theta),\dots,l_{u_{m}}(\theta)\right\} $
be the sample. We will consider Simple Random Sampling (SRS) which
means that
\[
\Pr(u_{i}=k)=\frac{1}{n}\text{ for }k=1,\dots,n,\text{ and for all }i=1,\dots,m.
\]
\citet{payne2015bayesian} use SRS (but without replacement) to form
an unbiased estimate of the log-likehood. However, the elements $l_{k}(\theta)$
(for a fixed $\theta$) vary substantially across the population and
estimating the total $\sum_{k=1}^{n}l_{k}(\theta)$ with SRS is well
known to be very inefficient under such a scenario. Instead $\Pr(u_{i}=k)$
should be (approximately) proportional to a size measure for $l_{k}(\theta)$.
\citet{quiroz2016speeding} argue that this so called Proportional-to-Size
sampling is in many cases unlikely to be successful as knowledge of
a size measure (which also depends on $\theta$) for all $k$ can
often defeat the purpose of subsampling. They instead propose to use
SRS but incorporate control variates in the estimator for variance
reduction which we now turn to.

\subsection{Efficient log-likelihood estimators\label{subsec:Efficient-log-likelihood-estimat}}

The idea in \citet{quiroz2016speeding} is to homogenize the population
$\left\{ l_{1}(\theta),\dots,l_{n}(\theta)\right\} $: if the resulting
elements are roughly of the same size then SRS is expected to be efficient.
Let $q_{k}(\theta)$ denote an approximation of $l_{k}(\theta)$ and
decompose
\begin{eqnarray}
l(\theta) & = & \sum_{k=1}^{n}q_{k}(\theta)+\sum_{k=1}^{n}\left[l_{k}(\theta)-q_{k}(\theta)\right]\label{eq:Compute_q}\\
 & = & q(\theta)+d(\theta),\nonumber 
\end{eqnarray}
where
\[
q(\theta)=\sum_{k=1}^{n}q_{k}(\theta),\quad d(\theta)=\sum_{k=1}^{n}d_{k}(\theta),\quad\text{and}\quad d_{k}(\theta)=l_{k}(\theta)-q_{k}(\theta).
\]
We emphasize that all quantities depend on $\theta$ which, from now
on, is sometimes suppressed for a compact notation. 

Define the random variables $\eta_{i}=nd_{u_{i}}$ and $d_{u_{i}}=l_{u_{i}}-q_{u_{i}}$
with $\mathrm{E}[\eta_{i}]=d$ and 
\[
\sigma_{\eta}^{2}=\mathrm{V}[\eta_{i}]=n^{2}\mathrm{V}[d_{u_{i}}]=n\sum_{k=1}^{n}(d_{k}-\bar{d})^{2},\quad\text{with }\bar{d}=\sum_{k=1}^{n}d_{k}.
\]
The difference estimator estimates $d$ in \eqref{eq:Compute_q} with
the Hansen-Hurwitz estimator \citep{hansen1943theory},
\[
\hat{d}_{m}=\frac{1}{m}\sum_{i=1}^{m}\eta_{i},\quad\text{with }\mathrm{E}[\hat{d}_{m}]=d\quad\text{and}\quad\mathrm{V}[\hat{d}_{m}]=\frac{\sigma_{\eta}^{2}}{m}.
\]
We can obtain an unbiased estimate $\hat{\sigma}_{\eta}^{2}=n^{2}\mathrm{\hat{V}}[d_{u_{i}}]$,
where $\hat{\mathrm{V}}[d_{u_{i}}]$ is the usual unbiased sample
variance estimator (of the population $\{d_{1},\dots d_{n}\}$). The
estimator of the log-likelihood is thus 
\begin{equation}
\hat{l}_{m}=\sum_{k=1}^{n}q_{k}(\theta)+\hat{d}_{m},\quad\text{with }\mathrm{E}[\hat{l}_{m}]=l\quad\text{and}\quad\sigma^{2}=\mathrm{V}[\hat{l}_{m}]=\frac{\sigma_{\eta}^{2}}{m}.\label{eq:log-likelihoodEstimator}
\end{equation}
It also follows that $\hat{\sigma}_{m}^{2}=\hat{\sigma}_{\eta}^{2}/m$
is an unbiased estimator of $\sigma^{2}$.

\citet{quiroz2016speeding} reason that observations close in data
space $(x_{k},y_{k})$ should, for a fixed $\theta$, have similar
$l_{k}(\theta)$ values. They cluster the data space into $K$ clusters
and approximate, within each cluster, $l_{k}(\theta)$ by a second
order Taylor series expansion $q_{k}(\theta)$ around the centroid
of the cluster. This allows computing $q$ using $K$ evaluations
(instead of $n$), see \citet{quiroz2016speeding} for details. \citet{bardenet2015markov}
propose similar control variates, but instead expand with respect
to $\theta$ around a reference value $\theta^{\star}$. While it
can be shown that $q$ can be computed using a single evaluation,
the approximation can be inaccurate when $\theta$ is far from $\theta^{\star}$
giving a large $\sigma^{2}$.

We note that the log-likelihood estimate in \citet{payne2015bayesian}
(if with replacement sampling is used) is a special case of \eqref{eq:log-likelihoodEstimator},
namely when $q_{k}=0$ for all $k$. We will see that this results
in a poor performance of the delayed acceptance algorithm and that
control variates are crucial for a successful implementation.

\subsection{Approximately bias-corrected likelihood estimator \label{subsec:Bias-corrected-likelihood}}

It is clear that an unbiased log-likelihood estimator becomes biased
for the likelihood when transformed to the ordinary scale by the exponential
function. Since by the Central Limit Theorem (CLT)
\[
\sqrt{m}\left(\hat{l}_{m}(\theta)-l(\theta)\right)\rightarrow\mathcal{N}(0,\sigma_{\eta}^{2})\quad\text{as}\quad m\rightarrow\infty,
\]
\citet{quiroz2016speeding} (see also \citealt{ceperley1999penalty,nicholls2012coupled})
approximately bias corrects the likelihood estimate 
\begin{equation}
\hat{p}_{m}(y|\theta,u)=\exp\left(\hat{l}_{m}(\theta)-\hat{\sigma}_{\eta}^{2}(\theta)/2m\right).\label{eq:LikelihoodEstimator}
\end{equation}
Equation \eqref{eq:LikelihoodEstimator} is unbiased if $\sigma_{\eta}^{2}$
is used in place of $\hat{\sigma}_{\eta}^{2}$ (and normality holds
for $\hat{l}_{m}$). In practice we need to use $\hat{\sigma}_{\eta}^{2}$
(to not use all data) and \citet{quiroz2016speeding} show that \eqref{eq:LikelihoodEstimator}
is asymptotically unbiased, with the bias decreasing as $O(m^{-2})$.
For the rest of the paper we refer to \eqref{eq:LikelihoodEstimator}
as a ``bias-corrected'' estimator, where the quotation marks highlight
that the correction is not exact.

\section{Delayed acceptance MH with subsets of data\label{sec:delayedMH}}

\subsection{The algorithm \label{subAlgorithm-DA-MH}}

\citet{payne2015bayesian} propose a subsampling delayed acceptance
MH following \citet{christen2005mcmc}. The aim in delayed acceptance
is to simulate a Markov chain $\{\theta^{(j)}\}_{j=1}^{N}$ which
admits the posterior
\[
\pi(\theta)=\frac{p(y|\theta)p(\theta)}{p(y)},\text{ where }p(y)=\int p(y|\theta)p(\theta)d\theta\text{ and }p(\theta)\text{ denotes the prior,}
\]
as invariant distribution. Moreover, the likelihood $p(y|\theta)$
should only be evaluated if there is a good chance of accepting the
proposed $\theta$. 

The algorithm in \citet{payne2015bayesian} proceeds as follows. Let
$\theta_{c}=\theta^{(j)}$ denote the current state of the Markov
chain. In the first stage, propose $\theta^{\prime}\sim q_{1}(\theta|\theta_{c})$
and compute
\begin{eqnarray}
\alpha_{1}(\theta_{c}\rightarrow\theta^{\prime}) & = & \min\left\{ 1,\frac{\hat{p}_{m}(y|\theta^{\prime},u)p(\theta^{\prime})/q_{1}(\theta^{\prime}|\theta_{c})}{\hat{p}_{m}(y|\theta_{c},u)p(\theta_{c})/q_{1}(\theta_{c}|\theta^{\prime})}\right\} ,\label{eq:alpha1}
\end{eqnarray}
where $\hat{p}_{m}(y|\theta,u)$ is the estimator in Section \ref{subsec:Bias-corrected-likelihood}
without control variates for $\hat{l}_{m}$ (but not ``bias-corrected'',
see below). Now, propose
\[
\theta_{p}=\begin{cases}
\theta^{\prime} & \quad\text{w.p.}\quad\alpha_{1}(\theta_{c},\theta^{\prime})\\
\theta_{c} & \quad\text{w.p.}\quad1-\alpha_{1}(\theta_{c},\theta^{\prime}),
\end{cases}
\]
and move the chain to the next state $\theta^{(j+1)}=\theta_{p}$
with probability
\begin{eqnarray}
\alpha_{2}(\theta_{c}\rightarrow\theta_{p}) & = & \min\left\{ 1,\frac{p(y|\theta_{p})p(\theta_{p})/q_{2}(\theta_{p}|\theta_{c})}{p(y|\theta_{c})p(\theta_{c})/q_{2}(\theta_{c}|\theta_{p})}\right\} ,\label{eq:alpha2}
\end{eqnarray}
where 
\begin{eqnarray*}
q_{2}(\theta_{p}|\theta_{c}) & = & \alpha_{1}(\theta_{c}\rightarrow\theta_{p})q_{1}(\theta_{p}|\theta_{c})+r(\theta_{c})\delta_{\theta_{c}}(\theta_{p}),\quad r(\theta_{c})=1-\int\alpha_{1}(\theta_{c}\rightarrow\theta_{p})q_{1}(\theta_{p}|\theta_{c})d\theta_{p},
\end{eqnarray*}
and $\delta$ is the Dirac delta function. If rejected we set $\theta^{(j+1)}=\theta_{c}$.

Note that $\alpha_{2}$ in \eqref{eq:alpha2} is equivalent to the
acceptance probability of a standard MH with a proposal density $q_{2}(\theta_{p}|\theta_{c})$
that is a mixture of two proposal densities: the first proposes to
move from $\theta_{c}$ to $\theta_{p}$ (from the ``slab'' $q_{1}$)
and the second proposes to stay at $\theta_{c}$ (from the ``spike'').
The mixture weight for the ``slab'' is $\alpha_{1}$ in \eqref{eq:alpha1}
(with $\theta^{\prime}=\theta_{p}$): if the likelihood estimate is
higher at $\theta^{\prime}$ compared to that of $\theta_{c}$ (after
correcting with $q_{1}$) we propose $\theta^{\prime}=\theta_{p}$
with probability 1. Conversely, if it is lower, we propose to move
but with a probability smaller than 1 which decreases the less likely
we think that $\theta^{\prime}$ will be accepted as indicated by
the estimated likelihood. The estimator $\hat{p}_{m}(y|\theta^{\prime},u)$
used by \citet{payne2015bayesian} does not depend on the current
state of the Markov chain and hence the mixture weights in $q_{2}$
are state-independent. Convergence to the invariant distribution therefore
follows from standard MH theory. The same applies when the estimator
uses the control variates in \citet{quiroz2016speeding}, or in \citet{bardenet2015markov}
but only for a fixed $\theta^{\star}$. \citet{bardenet2015markov}
suggest setting $\theta^{\star}=\theta_{c}$ every now and then to
prevent that the control variates can be poor if the chain is far
from $\theta^{\star}$: the resulting approximation is clearly state-dependent
and standard MH theory does not apply. Instead, convergence to the
invariant $\pi(\theta)$ is proved in \citet{christen2005mcmc} and
the delayed algorithm is exactly as above but with
\[
\hat{p}_{m}(y|\cdot,u)=\hat{p}_{m}^{(\theta_{c})}(y|\cdot,u),\quad\text{emphasizing that it depends on the current state }\theta_{c}.
\]
The state-dependent algorithm is a key ingredient when developing
the delayed acceptance block PMMH in Section \ref{sec:delayedPMMH}.

As noted above, \citet{payne2015bayesian} do not ``bias correct''
the likelihood estimate as they point out that the algorithm will
have the correct invariant distribution anyway. We remark that without
control variates it is not a good idea to apply the correction as
the variance of $\hat{\sigma}_{m}^{2}$ is huge which adversely affects
$\hat{p}_{m}(y|\cdot,u)$. An efficient estimate of $\hat{\sigma}_{m}^{2}$,
however, would certainly improve $\hat{p}_{m}(y|\cdot,u)$. While
the control variates allow us to estimate $\hat{\sigma}_{m}^{2}$
accurately, we will not implement the ``bias-correction'' when comparing
to \citet{payne2015bayesian} in order to make the comparison fair.

It is beneficial to also update $u$ in order to avoid the risk of
having a subset of observations for which the approximation is poor.
This is especially important for the estimator in \citet{payne2015bayesian},
as not using control variates can result in the particular subset
having highly heterogeneous elements, which is detrimental for SRS.
We note that updating $u$ is still a valid MH because (i) $u$ is
not a state of the Markov chain and (ii) the distribution $p(u)=1/n^{m}$
(SRS) does not depend on $\theta$. Thus, the transition kernel of
an algorithm that updates $u$ is a state-independent mixture of transition
kernels, where each of the kernels satisfies detailed balance either
by standard MH (or \citet{christen2005mcmc} if the approximation
depends on $\theta_{c}$). Since the weights $1/n^{m}$ do not depend
on $\theta$, it follows that the mixture also satisfies detailed
balance, and thus has $\pi(\theta)$ as invariant distribution. It
is unnecessary to update $u$ in every iteration, instead one can
update $u$ randomly to save the overhead cost of indexing the data
matrix when obtaining the subset of observations.

We remark that delayed acceptance (sometimes with names as early rejection
or surrogate transition), although without data subsampling, has been
considered earlier in the literature. References include \citet{fox1997sampling,liu2008monte,cui2011bayesian,smith2011estimating,solonen2012efficient,golightly2015delayed,sherlock2015adaptive}.
Each stage in \citet{banterle2014accelerating} (see Section \ref{sec:Introduction})
uses a partition of the data and can thus be considered as delayed
acceptance with data subsampling. The advantage of our algorithm is
that, because it only has two stages and the second stage evaluates
the full data likelihood, we can instead estimate this likelihood
in order to never do the evaluation for the full data set, see Section
\ref{sec:delayedPMMH}.

\subsection{Efficiency of delayed acceptance MH when subsampling the data}

When considering efficiency of MCMC algorithms with additional computational
costs (e.g. estimating the target), there are two types of fundamentally
different efficiencies that interplay. The first is the statistical
efficiency, which we will measure by the asymptotic (as the number
of MCMC iterates go to infinity) variance of an estimate based on
output from the Markov chain. Consider two MCMC algorithms $\mathcal{A}_{1}$
and $\mathcal{A}_{2}$ with the same invariant distribution. Then
$\mathcal{A}_{1}$ is statistically more (less) efficient than $\mathcal{A}_{2}$
if it has a lower (higher) asymptotic variance. The second is the
computational efficiency, which solely concerns ``execution time''
to produce a given number of iterates. The measures of statistical
and computational efficiency (and a combination of them) used in this
article are presented later in this section.

\citet{sherlock2015efficiency} (in a non-subsampling context) study
the statistical efficiency for delayed acceptance random walk Metropolis
and, moreover, an efficiency that also takes into account the computational
efficiency for the case where the target is estimated (DA-PMMH in
Section \ref{sec:delayedPMMH}).

\citet{christen2005mcmc} note that, because the transition kernels
of both the MH and delayed acceptance MH are derived from the same
proposal $q_{1}$, and in addition $\alpha_{2}\leq1$, the delayed
acceptance MH will be less statistically efficient than MH. The intuition
is that under these conditions the chain clearly exhibits a more ``sticky''
behavior and an estimate based on these samples will have a larger
asymptotic variance under DA-MH than MH. Notice that the closer $\alpha_{2}$
is to $1$, the more statistically efficient the delayed acceptance
algorithm is, and when $\alpha_{2}=1$ it is equivalent to the standard
MH which gives the upper bound of the possible statistical efficiency
achieved by a DA-MH.

Result 1 in \citet{payne2015bayesian} gives the alternative formulation
(for state-independent approximations) 
\begin{eqnarray}
\alpha_{2}(\theta_{c}\rightarrow\theta_{p}) & = & \min\left\{ 1,\frac{\hat{p}_{m}(y|\theta_{c},u)/p(y|\theta_{c})}{\hat{p}_{m}(y|\theta_{p},u)/p(y|\theta_{p})}\right\} .\label{eq:alpha2final}
\end{eqnarray}
 Let $l_{k}(\theta_{c},\theta_{p})=l_{k}(\theta_{c})-l_{k}(\theta_{p})$
and denote by $\hat{l}_{m}(\theta_{c},\theta_{p})$ the estimate of
$l(\theta_{c},\theta_{p})=\sum_{k=1}^{n}l_{k}(\theta_{c},\theta_{p})$.
Similarly to \eqref{eq:log-likelihoodEstimator},
\begin{eqnarray}
\hat{l}_{m}(\theta_{c},\theta_{p}) & = & q(\theta_{c},\theta_{p})+\frac{1}{m}\sum_{i=1}^{m}\zeta_{i},\quad\text{with }q(\theta_{c},\theta_{p})=\sum_{k=1}^{n}q_{k}(\theta_{c},\theta_{p}),\label{eq:diff_likelihood_est}
\end{eqnarray}
where $q_{k}(\theta_{c},\theta_{p})=q_{k}(\theta_{c})-q_{k}(\theta_{p})$
and the $\zeta_{i}$'s are iid with
\[
\Pr\left(\zeta_{i}=nD_{k}(\theta_{c},\theta_{p})\right)=\frac{1}{n},\text{ with }D_{k}=\left(l_{k}(\theta_{c},\theta_{p})-q_{k}(\theta_{c},\theta_{p})\right)\quad\text{for }i=1,\dots m.
\]
We can also show that
\begin{equation}
\mathrm{E}[\hat{l}_{m}(\theta_{c},\theta_{p})]=l(\theta_{c},\theta_{p})\text{ and }V[\hat{l}_{m}(\theta_{c},\theta_{p})]=\frac{\sigma_{\zeta}^{2}}{m}\text{ with }\sigma_{\zeta}^{2}=n\sum_{k=1}^{n}\left(D_{k}(\theta_{c},\theta_{p})-\bar{D}_{F}(\theta_{c},\theta_{p})\right)^{2},\label{eq:mean_var_log_ratio}
\end{equation}
where $\bar{D}_{F}$ is the mean over the full population. When not
``bias-corrected'', the ratio appearing in \eqref{eq:alpha2final}
becomes
\begin{equation}
R_{m}=\exp\left(\hat{l}_{m}(\theta_{c},\theta_{p})-l(\theta_{c},\theta_{p})\right).\label{eq:assymptotic_dist_ratios-1-1}
\end{equation}

We now propose a theorem that relates $\mathrm{E}\left[\alpha_{2}(\theta_{c}\rightarrow\theta_{p})\right]$
to the variance $\sigma_{R}^{2}=\mathrm{V}[\hat{l}_{m}(\theta_{c},\theta_{p})]$
under the assumption that $\hat{l}_{m}(\theta_{c},\theta_{p})\sim\mathcal{N}(l(\theta_{c},\theta_{p}),\sigma_{R}^{2})$
(equivalently $R_{m}\sim\log\mathcal{N}(0,\sigma_{R}^{2})$ in \eqref{eq:assymptotic_dist_ratios-1-1}).
In turn, $\alpha_{2}$ relates to the statistical efficiency as discussed
above. The assumption of normality is justified by a standard CLT
for $\hat{l}_{m}(\theta_{c},\theta_{p})$ since the $\zeta_{i}$'s
are iid.
\begin{thm}
\label{thm:theorem1}Suppose that we run a DA-MH with a state-independent
approximation
\[
\hat{l}_{m}(\theta_{c},\theta_{p})\sim\mathcal{N}\left(l(\theta_{c},\theta_{p}),\sigma_{R}^{2}\right),\quad\text{where }\sigma_{R}^{2}(\theta_{c},\theta_{p})=\mathrm{V}[\log(R_{m})],
\]
which has a second stage acceptance probability
\[
\alpha_{2}(\theta_{c}\rightarrow\theta_{p})=\min\left(1,R_{m}\right),\quad R_{m}=\exp\left(\hat{l}_{m}(\theta_{c},\theta_{p})-l(\theta_{c},\theta_{p})\right).
\]
Then
\begin{eqnarray*}
\mathrm{E}[\alpha_{2}(\theta_{c}\rightarrow\theta_{p})] & = & \exp\left(\sigma_{R}^{2}(\theta_{c},\theta_{p})/2\right)\left(1-\Phi\left(\sigma_{R}(\theta_{c},\theta_{p})\right)\right)+0.5.
\end{eqnarray*}
In particular, $\mathrm{E}[\alpha_{2}(\theta_{c}\rightarrow\theta_{p})]$
is a decreasing function of $\sigma_{R}$.
\end{thm}
\begin{proof}
See Appendix \ref{sec:Comparing-the-Efficiency}.
\end{proof}
\begin{rem*}
It is possible to state and prove a similar theorem using a ``bias-corrected''
likelihood estimator, which in log-scale is $\hat{l}_{m}(\theta_{c},\theta_{p})-\hat{\sigma}{}_{R}^{2}(\theta_{c},\theta_{p})/2$.
This is omitted as our application in Example 1 in Section \ref{sec:Application}
does not use a ``bias-corrected'' likelihood estimator in order
to conduct a fair comparison to \citet{payne2015bayesian}.
\end{rem*}
Theorem \ref{thm:theorem1} says that if the variance of the log of
the ratio $R_{m}$ in \eqref{eq:assymptotic_dist_ratios-1-1} is lower,
then the algorithm has a higher $\alpha_{2}$ on average. Moreover,
it shows that $\alpha_{2}$ deteriorates quickly with $\sigma_{R}$,
which illustrates the importance of achieving a small $\sigma_{R}$.

Although the delayed acceptance MH is always less statistically efficient
than MH it can of course still be more generally efficient in terms
of balancing computational and statistical efficiency. Clearly, a
necessary condition to achieve a more general efficient DA-MH algorithm
than the corresponding MH requires that its computing time is faster,
i.e. it must be computationally more efficient. When comparing DA-MH
algorithms with the same computing time (by for example having the
same subsample size) then Theorem \ref{thm:theorem1} shows that a
smaller variance of the log-ratio results in a more general efficient
algorithm.

To also compare algorithms of different computing times we now define
a measure for the general efficiency discussed above. In the rest
of the paper, whenever we claim that algorithms are more or less efficient
to each other it is based on this measure. The statistical (in)efficiency
part is measured by the Inefficiency Factor (IF), which quantifies
the amount by which the variance of $(1/N)\sum_{j=1}^{N}\theta^{(j)}$
is inflated when $\theta^{(j)}$ is obtained by Markov chain simulation
compared to that of iid simulation. It is given by
\begin{eqnarray}
\mathrm{IF} & = & 1+2\sum_{l=1}^{\infty}\rho_{l},\label{eq:IF}
\end{eqnarray}
where $\rho_{l}$ is the auto-correlation at the $l$th lag of the
chain, and can be computed with the coda package in R \citep{coda2006}.
To include computational efficiency in the measure we use Effective
Draws (ED) per computing unit
\begin{eqnarray}
\mathrm{ED} & = & \frac{N}{\mathrm{IF}\times t},\label{eq:ED}
\end{eqnarray}
where $N$ is the number of MCMC iterations and $t$ is the computing
time. The measure of interest is the effective draws per computing
time of a delayed acceptance algorithm $\mathcal{A}$ relative to
that of MH, i.e.
\begin{eqnarray}
\mathrm{RED} & = & \frac{\mathrm{ED}^{\mathcal{A}}}{\mathrm{ED}^{\mathrm{MH}}}.\label{eq:RED}
\end{eqnarray}

\section{Delayed acceptance PMMH with subsets of data\label{sec:delayedPMMH}}

\subsection{PMMH for data subsampling\label{subsec:PMMH}}

\citet{quiroz2016speeding} propose a pseudo-marginal approach to
subsampling based on the ``bias-corrected'' likelihood estimator
in \eqref{eq:LikelihoodEstimator}. In the next subsection this estimate
replaces the true likelihood, thereby avoiding the evaluation of the
full data set. To allow for a smaller subsample size without adversely
affecting the mixing of the chain \citet{quiroz2016speeding} develop
a correlated pseudo-marginal approach based on \citet{deligiannidis2015correlated}.
\citet{tran2016block} \citep[see also][]{quiroz2016speeding} use
an alternative approach to correlated pseudo-marginal which we use
for our Delayed Acceptance Block PMMH (DA-BPMMH) in Section \ref{subsec:DA-CPPMH}.
Although not considered here, it is straightforward to instead correlate
the subsamples as in \citet{deligiannidis2015correlated}.

Pseudo-marginal by data subsampling targets the following posterior
on an augmented space $(\theta,u)$, 
\begin{equation}
\tilde{\pi}_{m}(\theta,u)=\hat{p}_{m}(y|\theta,u)p(u)p(\theta)/p_{m}(y),\text{ with }p_{m}(y)=\int p_{m}(y|\theta)p(\theta)d\theta.\label{eq:AugmentedPosterior-1-1}
\end{equation}
The algorithm is similar to MH except that $\theta$ and $u$ are
proposed and accepted (or rejected) jointly, with probability
\begin{eqnarray}
\alpha_{\mathrm{PMMH}} & = & \min\left\{ 1,\frac{\hat{p}_{m}(y|\theta_{p},u_{p})p(\theta_{p})/q(\theta_{p}|\theta_{c})}{\hat{p}_{m}(y|\theta_{c},u_{c})p(\theta_{c})/q(\theta_{c}|\theta_{p})}\right\} ,\label{eq:alphaPMMH}
\end{eqnarray}
where the subscripts denote the current and proposed values of $\theta$
and $u$. \citet{quiroz2016speeding} show that the algorithm converges
to a slightly perturbed target (only approximately unbiased likelihood
estimator) $\pi_{m}(\theta)$ and prove that
\[
\frac{\left|\pi_{m}(\theta)-\pi(\theta)\right|}{\pi(\theta)}\leq O(m^{-2}).
\]
Moreover, if $h(\theta)$ is a function that is absolutely integrable
with respect to $\pi(\theta)$, then $\mathrm{E}_{\pi_{m}}\left[h(\theta)\right]$
is also within $O(m^{-2})$ of its true value.

The variance of $\hat{l}_{m}$ is crucial for the general efficiency
of the algorithm: $\sigma^{2}$ in the approximate interval $[1,3.3]$
is optimal \citep{pitt2012some,doucet2015efficient,sherlock2015pseudo}.
The correlated pseudo-marginal approach induces a high positive correlation
between the estimates in \eqref{eq:alphaPMMH} by correlating $u_{c}$
and $u_{p}$. This allows the use of a less precise estimator without
getting stuck: the errors in the numerator and denominator tend to
cancel. As a consequence, $\sigma^{2}$ can be larger than above,
hence speeding up the algorithm by taking a smaller subsample. We
follow \citet{tran2016block} and divide $u$ (defined in Section
\ref{subsec:Data-subsampling}) into $G$ blocks,
\[
u=(u_{(1)},\dots,u_{(G)})\quad\text{with }\frac{m}{G}\text{ observations within each block,}
\]
and update a single block (randomly) in each iteration. Setting $G$
large induces a high positive correlation $\rho$ between $\hat{l}_{m}(\theta_{c})$
and $\hat{l}_{m}(\theta_{p})$: \citet{tran2016block} show that $\rho=1-1/G$
under certain assumptions. We set $G=100$ for the state-dependent
algorithm in our application.

\subsection{State-independent algorithm: DA-PMMH}

Delayed acceptance PMMH uses an estimate of the target in the second
stage of the algorithm. Such algorithms have recently been considered
by \citet{sherlock2015efficiency,sherlock2015adaptive} and \citet{golightly2015delayed}.
We propose a state-independent data subsampling DA-PMMH (uncorrelated
PMMH) and a state-dependent (block PMMH, where the correlation between
the subsamples follows \citealp{tran2016block}) extension DA-Block
(correlated) PMMH (DA-BPMMH) in the next subsection.

The estimated ``bias-corrected'' likelihood in \eqref{eq:LikelihoodEstimator}
is approximated in a first stage screening by $\hat{s}$ which is
discussed in Section \ref{subsec:Approx-EstimatedLikelihood}. Similarly
to the algorithm in Section \ref{subAlgorithm-DA-MH}, we desire to
only evaluate the (estimated) target (on the augmented space) if the
proposed state is likely to be accepted. Let $\left(\theta_{c},u_{c}\right)=\left(\theta^{(j)},u^{(j)}\right)$
denote the current state of the augmented Markov chain. In the first
stage, propose $\theta^{\prime}\sim q_{1}(\theta|\theta_{c})$ and
$u^{\prime}\sim p(u)$ and evaluate
\begin{eqnarray}
\alpha_{1}^{\mathrm{PMMH}}\left\{ \left(\theta_{c},u_{c}\right)\rightarrow\left(\theta^{\prime},u^{\prime}\right)\right\}  & = & \min\left\{ 1,\frac{\hat{s}(\theta^{\prime},u^{\prime})p(\theta^{\prime})/q_{1}(\theta^{\prime}|\theta_{c})}{\hat{s}(\theta_{c},u_{c})p(\theta_{c})/q_{1}(\theta_{c}|\theta^{\prime})}\right\} ,\label{eq:alpha1-1}
\end{eqnarray}
where $\hat{s}(\theta,u)$ is the approximation of $\hat{p}_{m}(y|\theta,u)$
\eqref{eq:LikelihoodEstimator}. Propose
\[
\left(\theta_{p},u_{p}\right)=\begin{cases}
\left(\theta^{\prime},u^{\prime}\right) & \quad\text{w.p.}\quad\alpha_{1}^{\mathrm{PMMH}}\left\{ \left(\theta_{c},u_{c}\right)\rightarrow\left(\theta^{\prime},u^{\prime}\right)\right\} \\
\left(\theta_{c},u_{c}\right) & \quad\text{w.p.}\quad1-\alpha_{1}^{\mathrm{PMMH}}\left\{ \left(\theta_{c},u_{c}\right)\rightarrow\left(\theta^{\prime},u^{\prime}\right)\right\} ,
\end{cases}
\]
and move to $\left(\theta^{(j+1)},u^{(j+1)}\right)=\left(\theta_{p},u_{p}\right)$
with probability
\begin{eqnarray}
\alpha_{2}^{\mathrm{PMMH}}\left\{ \left(\theta_{c},u_{c}\right)\rightarrow\left(\theta_{p},u_{p}\right)\right\}  & = & \min\left\{ 1,\frac{\hat{p}_{m}(y|\theta_{p},u_{p})p(\theta_{p})/q_{2}(\theta_{p}|\theta_{c})}{\hat{p}_{m}(y|\theta_{c},u_{c})p(\theta_{c})/q_{2}(\theta_{c}|\theta_{p})}\right\} ,\label{eq:alpha2-1}
\end{eqnarray}
where 
\begin{eqnarray*}
q_{2}(\theta_{p}|\theta_{c}) & = & \alpha_{1}^{\mathrm{PMMH}}q_{1}(\theta_{p}|\theta_{c})+r(\theta_{c})\delta_{\theta_{c}}(\theta_{p}),\quad r(\theta_{c})=1-\int\alpha_{1}^{\mathrm{PMMH}}q_{1}(\theta_{p}|\theta_{c})d\theta_{p},
\end{eqnarray*}
and $\delta$ is the Dirac delta function. If rejected we set $\left(\theta^{(j+1)},u^{(j+1)}\right)=\left(\theta_{c},u_{c}\right)$.
Similarly to the argument for the DA-MH, we recognize $\alpha_{2}^{\mathrm{PMMH}}$
as the acceptance probability for a pseudo-marginal algorithm. Since
the approximation is state-independent ($\hat{s}(\theta,u)$ independent
of the current state implies state-independent mixture weights for
$q_{2}$), convergence follows from being a pseudo-marginal algorithm
\citep{andrieu2009pseudo}. However, as the estimate is biased the
target is perturbed \citep{quiroz2016speeding} but within $O(m^{-2})$
as discussed in Section \ref{subsec:PMMH}.

\subsection{State-dependent algorithm: DA-BPMMH\label{subsec:DA-CPPMH}}

As $u$ is part of the state in a pseudo-marginal algorithm, obtaining
a state-dependent approximation is easily achieved by correlating
$u$: we sample $u_{p}\sim p(u|u_{c})$ thus $\hat{s}(\theta,u)=\hat{s}_{\theta_{c},u_{c}}(\theta,u)$.
This is the state-dependent (on the augmented space) setup in \citet{christen2005mcmc}
and it follows that the invariant distribution is $\tilde{\pi}_{m}(\theta,u)$
in \eqref{eq:AugmentedPosterior-1-1}. This is the invariant distribution
in the algorithm in \citet{quiroz2016speeding}, which we already
mentioned has a marginal $\pi_{m}(\theta)$ within $O(m^{-2})$ of
$\pi(\theta)$. We note that this state-dependent approximation is
very convenient because it automatically allows us to have a higher
variance on $\hat{l}_{m}$ because of the block PMMH mechanism (see
Section \ref{subsec:PMMH}).

\subsection{An approximation of the estimated likelihood\label{subsec:Approx-EstimatedLikelihood}}

Our approximation is inspired by adaptive delayed acceptance ideas
in \citet{cui2011bayesian} and \citet{sherlock2015adaptive}. However,
our approach is not strictly adaptive as we only learn about the proposal
for a fixed training period of $N_{\mathrm{train}}$ iterations, which
are discarded from the final draws.

The idea is to use a sparser set of the data to construct control
variates $q_{k}^{(1)}(\theta)$ in the first stage. During the training
period we learn about the discrepancy between $q^{(1)}(\theta)=\sum_{k=1}^{n}q_{k}^{(1)}(\theta)$
and $q(\theta)$ of the second stage (obtained with a denser set),
i.e. 
\begin{equation}
\underbrace{q(\theta)-q^{(1)}(\theta)}_{=e(\theta)}=f(\theta)+\epsilon,\quad\epsilon\text{ is the noise (assumed independent of \ensuremath{\theta})}.\label{eq:discrepency_between_stages}
\end{equation}
Learning about $f$ is a standard regression problem: the collection
of proposed $\theta$'s during the fixed training period are the inputs
and the discrepancies are the training data.

The trivial decomposition
\begin{eqnarray*}
\hat{l}_{m}(\theta) & = & \sum_{k=1}^{n}q_{k}(\theta)+\hat{d}_{m}(\theta)-\hat{\sigma}_{m}^{2}/2\\
 & = & \sum_{k=1}^{n}q_{k}^{(1)}(\theta)+\left(\sum_{k=1}^{n}q_{k}(\theta)-\sum_{k=1}^{n}q_{k}^{(1)}(\theta)\right)+\hat{d}_{m}(\theta)-\hat{\sigma}_{m}^{2}/2
\end{eqnarray*}
suggests the first stage approximation
\[
\hat{s}(\theta,u)=\sum_{k=1}^{n}q_{k}^{(1)}(\theta)+\hat{e}(\theta)+\hat{d}_{m}(\theta)-\hat{\sigma}_{m}^{2}/2,
\]
where $\hat{e}(\theta)$ is the prediction of the discrepancy at $\theta$.
Note that computing the true discrepancy requires the denser set to
be evaluated, whereas the prediction is very fast. For example, in
linear (or non-linear by basis functions) regression the parameters
of $f(\theta)$ are estimated once after the training data has been
collected. Prediction of $e(\theta)$ for a new ``data-observation''
$\theta$ is then a simple dot product computation, which is typically
much faster than evaluating $q(\theta)$. Note that if the $u$'s
are correlated then $\hat{s}(\theta,u)$ also depends on the current
state, i.e. $\hat{s}_{\theta_{c},u_{c}}(\theta,u)$.

In Section \ref{sec:Application} we implement both a linear regression
and (noise free) Gaussian process to learn $f$ in \eqref{eq:discrepency_between_stages},
but any regression technique can be used.

\section{Application\label{sec:Application}}

\subsection{Data and model}

We model the probability of bankruptcy conditional on a set of covariates
using a data set of $534,717$ Swedish firms for the time period 1991-2008.
We have in total $n=4,748,089$ firm-year observations. The variables
included are: earnings before interest and taxes, total liabilities,
cash and liquid assets, tangible assets, logarithm of deflated total
sales and logarithm of firm age in years. We also include the macroeconomic
variables GDP-growth rate (yearly) and the interest rate set by the
Swedish central bank. See \citet{giordani2013taking} for a detailed
description of the data set. 

We consider the logistic regression model
\begin{eqnarray*}
p(y_{k}|x_{k},\beta) & = & \left(\frac{1}{1+\exp(x_{k}^{T}\beta)}\right)^{y_{k}}\left(\frac{1}{1+\exp(-x_{k}^{T}\beta)}\right)^{1-y_{k}},
\end{eqnarray*}
where $x_{k}$ includes the variables above plus an intercept term.
We set $p(\beta)\sim N(0,10I)$ for simplicity. Since the bankruptcy
observations $(y_{k}=1)$ are sparse in the data, we follow \citet{payne2015bayesian}
and estimate the likelihood only for the $y_{k}=0$ observations.
That is, we decompose
\begin{eqnarray*}
l(\beta) & = & \sum_{\{k;y_{k}=1\}}l_{k}(\beta)+\sum_{\{k;y_{k}=0\}}l_{k}(\beta),
\end{eqnarray*}
and evaluate the first term whereas a random sample is only taken
to estimate the second term. The second term clearly follows the structure
presented in Section \ref{subsec:Structure-likelihood}.

\subsection{Performance evaluation measures}

The quantity of interest is the effective draws as defined in \eqref{eq:ED}.
We now outline how to compute $t$ as CPU time and present an alternative
measure independent of the implementation that is based number of
evaluations. We first outline a robust measure of the CPU time.

The delayed acceptance algorithms we implement have two stages, where
the first stage is filtering out draws unlikely to be accepted. One
can view MH as an algorithm that does not filter any proposed draws
at all: any draw is subject to an accept/reject decision based on
the second stage likelihood. To make total CPU time comparisons fair
between a Delayed Acceptance (DA) MH and its corresponding MH, we
compute the total time for the latter (MH) by the median time the
former (DA) spends in the second stage and multiply by the number
of MCMC iterations $N$: 
\[
\mathrm{CPU}_{\mathrm{MH}}=N\times\text{median time \ensuremath{\mathrm{DA}} stage 2}.
\]
The median is used to avoid extreme values that can arise due to external
disturbances (CPU time should be independent of $\theta$ here). For
the delayed acceptance algorithm the CPU time is
\[
\mathrm{CPU}_{\mathrm{DA}}=N\times\text{median time \ensuremath{\mathrm{DA}} stage 1}\text{ }+\mathrm{FullEval\times\text{median time \ensuremath{\mathrm{DA}} stage 2},}
\]
where $\mathrm{FullEval}$ is the number of second stage evaluations
the delayed acceptance performs. 

The (total) number of evaluations measure is straightforward for MH
($N\times n$), DA-MH without control variates ($N\times m+\mathrm{FullEval}\times n$)
and with ($N\times(K+m)+\mathrm{FullEval}\times n$), and PMMH/BPMMH
($N\times(K+m)$). For the delayed versions of PMMH/BPMMH we similarly
add the different evaluations, but note that it will be different
for the pre- and post- training period. Moreover, there is a one time
cost of learning $f(\theta)$ and also a (post-training) cost of predicting
$e(\theta)$. We translate these to number of evaluations as follows.
First, for learning $f(\theta)$ we measure the CPU time it takes
to fit it with the training data. We compare this time to the average
CPU time for the iterations during the training period. We do so because
we know the measure of the latter in terms of number of evaluations:
$K^{(1)}+K+m$, where $K^{(1)}$ and $K$ are, respectively, the number
of centroids in the sparser and denser set of data. Therefore, if
learning $f(\theta)$ is say $T$ times slower in CPU time, then this
translates to $T\cdot(K^{(1)}+K+m)$ number of evaluations. Finally,
the number of evaluations for predicting a single $e(\theta)$ depends
on the model for $f(\theta)$. For linear regression the prediction
is a dot product which we assign the same cost as computing a single
log-likelihood contribution (which is typically a function of a dot
product). For a Gaussian process, the prediction requires evaluating
the kernel for $N_{\mathrm{train}}$ observations and we let that
define the number of evaluations for a single prediction.

\subsection{Implementation details\label{subsec:Implementation-details}}

We consider the following two examples. Example 1 estimates the model
with DA-MH implemented with our efficient control variates and compares
it to the implementation in \citet{payne2015bayesian}. Example 2
estimates the model with DA-BPMMH and compare to the block PMMH algorithm
in \citet{quiroz2016speeding}. To make comparisons fair for Example
1 we use without replacement sampling (as in \citealp{payne2015bayesian}).
This sampling scheme is typically used together with the Horvitz-Thompson
estimator \citep{horvitz1952generalization}: see \citet{sarndal2003model}
on how to modify the formulas in Section \ref{subsec:Efficient-log-likelihood-estimat}
for without replacement sampling. 

Two main implementations of the difference estimator are considered.
The first computes $q_{k}$ with the second order term evaluated at
$\beta$, which we call \emph{dynamic}. The second, which we call
\emph{static}, fixes the second order term at the optimum $\beta^{\star}$.
The dynamic approach clearly provides a better approximation but is
more expensive to compute. For both the dynamic and static approaches
we compare four different sparse representations of the data for computing
$q$ in \eqref{eq:Compute_q}, each with a different number of clusters.
The clusters are obtained using Algorithm 1 in \citet{quiroz2016speeding}
on the observations for which $y=0$ ($4,706,523$ observations).
We note that, as more clusters are used to represent the data, the
approximation of the likelihood is more accurate, although it is more
expensive to compute.

We consider a Random walk MH proposal for $\beta$ where we learn
the proposal scale during the first $N_{\mathrm{train}}=5,000$ (and
also train $f(\theta)$) iterations in order to reach an acceptance
probability of $\approx0.23$ for MH \citep{roberts1997weak} and
$\approx0.10$ for BPMMH. For the delayed acceptance algorithms we
have the same targets but for $\alpha_{1}$, i.e. the first stage
acceptance probability. We discard the training samples and also a
subsequent burn-in period of $10$\% of the remaining samples ($20,000$)
when doing inference. However, the computing costs (CPU and number
of evaluations) include all $N$ iterations. 

Finally, the delayed acceptance algorithms are implemented with an
update of $u$ with probability $0.01$.

\subsection{Example 1: DA-MH\label{subsec:Results}}

Tables \ref{tab:DE} and \ref{tab:PM} summarize the results, respectively,
for the difference estimator with control variates (DE) and the estimator
in \citet{payne2015bayesian} (PM). It is evident that the difference
estimator has a larger second stage acceptance probability $\alpha_{2}$
(for a given sample size), which is a consequence of Theorem \ref{thm:theorem1}
because it has a lower $\sigma_{R}^{2}=V[\log(R_{m})].$ Figure \ref{fig:NormalityExample1}
confirms that the normality assumptions, for the smallest value of
$m$ and $K$ (the worst case scenario), are adequate for both methods.
We also note from Table \ref{tab:PM} that for some sample sizes \citet{payne2015bayesian}
performs more poorly than the standard Metropolis-Hastings algorithm.
One possible explanation is that the applications in \citet{payne2015bayesian}
have a small number of continuous covariates (one in the first application
and three in the second) and the rest are binary. It is clear that
the continuous covariate case results in more variation among the
log-likelihood contributions which is detrimental for SRS. In this
application we have eight continuous covariates which explains why
SRS without covariates performs poorly for small sampling fractions.
As an example, for a subsample of $0.1\%$ of the data, not a single
effective sample was obtained.

\begin{table}[h]
\centering \caption{\emph{Delayed acceptance MH with control variates.} The table shows
some quantities for the static and dynamic implementation with different
sparse representations of the data represented by $K$, which is the
number of clusters (expressed as \% of $n$). For each approximation
different sample sizes ($0.1,1,5$ in \% of $n$) are considered.
The quantities are the mean $\mathrm{RED}_{1}$ and $\mathrm{RED}_{2}$
in \eqref{eq:RED} measured with respect to computing time and average
number of evaluations, respectively. Furthermore, $\bar{\sigma}_{R}$
is the mean (over MCMC iterations) standard deviation of $\log(R_{m})$.
Finally, $\alpha_{1}$ and $\alpha_{2}$ are the acceptance probabilities
in \eqref{eq:alpha1} and \eqref{eq:alpha2final} (expressed in \%),
where the latter is computed conditional on acceptance in the first
stage. The corresponding MH algorithm has an acceptance rate of $\approx$$23$\%.}

{\footnotesize{}}%
\begin{tabular}{llrrrrrrrrrrrrrrrrrrrc}
\toprule 
 &  & \multicolumn{4}{l}{\textbf{\footnotesize{}Static}} & \multicolumn{3}{c}{} &  &  &  & \multicolumn{3}{l}{\textbf{\footnotesize{}Dynamic}} &  &  &  &  &  &  & \tabularnewline
\cmidrule{3-11} \cmidrule{13-21} 
 &  & {\footnotesize{}$\mathrm{RED}_{1}$} &  & {\footnotesize{}$\mathrm{RED}_{2}$} &  & {\footnotesize{}$\bar{\sigma}_{R}$} &  & {\footnotesize{}$\alpha_{1}$} &  & {\footnotesize{}$\alpha_{2}$} &  & {\footnotesize{}$\mathrm{RED}_{1}$} &  & {\footnotesize{}$\mathrm{RED}_{2}$} &  & {\footnotesize{}$\bar{\sigma}_{R}$} &  & {\footnotesize{}$\alpha_{1}$} &  & {\footnotesize{}$\alpha_{2}$} & \tabularnewline
\cmidrule{3-3} \cmidrule{5-5} \cmidrule{7-7} \cmidrule{9-9} \cmidrule{11-11} \cmidrule{13-13} \cmidrule{15-15} \cmidrule{17-17} \cmidrule{19-19} \cmidrule{21-21} 
 & \textbf{\footnotesize{}$K=0.03$ } &  &  &  &  &  &  &  &  &  &  &  &  &  &  &  &  &  &  &  & \tabularnewline
\cmidrule{2-2} 
 & {\footnotesize{}0.1} & {\footnotesize{}0.93} &  & {\footnotesize{}0.95} &  & {\footnotesize{}5.90} &  & {\footnotesize{}22} &  & {\footnotesize{}14} &  & {\footnotesize{}2.35} &  & {\footnotesize{}2.56} &  & {\footnotesize{}1.48} &  & {\footnotesize{}23} &  & {\footnotesize{}57} & \tabularnewline
 & {\footnotesize{}1} & {\footnotesize{}2.43} &  & {\footnotesize{}2.65} &  & {\footnotesize{}1.40} &  & {\footnotesize{}22} &  & {\footnotesize{}59} &  & {\footnotesize{}2.23} &  & {\footnotesize{}3.47} &  & {\footnotesize{}0.45} &  & {\footnotesize{}22} &  & {\footnotesize{}85} & \tabularnewline
 & {\footnotesize{}5} & {\footnotesize{}2.09} &  & {\footnotesize{}2.78} &  & {\footnotesize{}0.57} &  & {\footnotesize{}24} &  & {\footnotesize{}82} &  & {\footnotesize{}0.73} &  & {\footnotesize{}3.27} &  & {\footnotesize{}0.19} &  & {\footnotesize{}24} &  & {\footnotesize{}93} & \tabularnewline
\cmidrule{2-2} 
 & \textbf{\footnotesize{}$K=0.21$ } &  &  &  &  &  &  &  &  &  &  &  &  &  &  &  &  &  &  &  & \tabularnewline
\cmidrule{2-2} 
 & {\footnotesize{}0.1} & {\footnotesize{}1.51} &  & {\footnotesize{}1.52} &  & {\footnotesize{}3.87} &  & {\footnotesize{}21} &  & {\footnotesize{}24} &  & {\footnotesize{}2.87} &  & {\footnotesize{}2.84} &  & {\footnotesize{}0.83} &  & {\footnotesize{}25} &  & {\footnotesize{}74} & \tabularnewline
 & {\footnotesize{}1} & {\footnotesize{}2.81} &  & {\footnotesize{}3.02} &  & {\footnotesize{}0.99} &  & {\footnotesize{}22} &  & {\footnotesize{}69} &  & {\footnotesize{}3.03} &  & {\footnotesize{}3.91} &  & {\footnotesize{}0.27} &  & {\footnotesize{}22} &  & {\footnotesize{}91} & \tabularnewline
 & {\footnotesize{}5} & {\footnotesize{}2.05} &  & {\footnotesize{}2.71} &  & {\footnotesize{}0.39} &  & {\footnotesize{}26} &  & {\footnotesize{}88} &  & {\footnotesize{}0.78} &  & {\footnotesize{}3.17} &  & {\footnotesize{}0.11} &  & {\footnotesize{}25} &  & {\footnotesize{}96} & \tabularnewline
\cmidrule{2-2} 
 & \textbf{\footnotesize{}$K=0.71$ } &  &  &  &  &  &  &  &  &  &  &  &  &  &  &  &  &  &  &  & \tabularnewline
\cmidrule{2-2} 
 & {\footnotesize{}0.1} & {\footnotesize{}2.24} &  & {\footnotesize{}2.20} &  & {\footnotesize{}2.10} &  & {\footnotesize{}21} &  & {\footnotesize{}46} &  & {\footnotesize{}3.03} &  & {\footnotesize{}3.49} &  & {\footnotesize{}0.41} &  & {\footnotesize{}23} &  & {\footnotesize{}86} & \tabularnewline
 & {\footnotesize{}1} & {\footnotesize{}3.53} &  & {\footnotesize{}3.30} &  & {\footnotesize{}0.60} &  & {\footnotesize{}22} &  & {\footnotesize{}81} &  & \textit{\emph{\footnotesize{}1.97}} &  & \textit{\emph{\footnotesize{}3.70}} &  & {\footnotesize{}0.13} &  & \textit{\emph{\footnotesize{}23}} &  & \textit{\emph{\footnotesize{}96}} & \tabularnewline
 & {\footnotesize{}5} & {\footnotesize{}2.22} &  & {\footnotesize{}3.09} &  & {\footnotesize{}0.26} &  & {\footnotesize{}22} &  & {\footnotesize{}92} &  & {\footnotesize{}0.70} &  & {\footnotesize{}3.28} &  & {\footnotesize{}0.06} &  & {\footnotesize{}24} &  & {\footnotesize{}98} & \tabularnewline
\cmidrule{2-2} 
 & \textbf{\footnotesize{}$K=3.68$ } &  &  &  &  &  &  &  &  &  &  &  &  &  &  &  &  &  &  &  & \tabularnewline
\cmidrule{2-2} 
 & {\footnotesize{}0.1} & {\footnotesize{}2.64} &  & {\footnotesize{}2.75} &  & {\footnotesize{}0.82} &  & {\footnotesize{}21} &  & {\footnotesize{}74} &  & {\footnotesize{}1.28} &  & {\footnotesize{}3.63} &  & {\footnotesize{}0.13} &  & {\footnotesize{}21} &  & {\footnotesize{}96} & \tabularnewline
 & {\footnotesize{}1} & \textit{\emph{\footnotesize{}3.71}} &  & \textit{\emph{\footnotesize{}3.19}} &  & \textit{\emph{\footnotesize{}0.25}} &  & \textit{\emph{\footnotesize{}23}} &  & \textit{\emph{\footnotesize{}92}} &  & {\footnotesize{}1.15} &  & {\footnotesize{}3.57} &  & {\footnotesize{}0.04} &  & {\footnotesize{}21} &  & {\footnotesize{}99} & \tabularnewline
 & {\footnotesize{}5} & {\footnotesize{}3.60} &  & {\footnotesize{}2.93} &  & {\footnotesize{}0.11} &  & {\footnotesize{}23} &  & {\footnotesize{}96} &  & {\footnotesize{}0.60} &  & {\footnotesize{}2.95} &  & {\footnotesize{}0.02} &  & {\footnotesize{}24} &  & {\footnotesize{}99} & \tabularnewline
\bottomrule
\end{tabular}\label{tab:DE}
\end{table}
\begin{table}[h]
\centering \caption{\emph{Delayed acceptance MH without control variates }\citep{payne2015bayesian}\emph{.}
The table shows some quantities for different sample sizes ($0.1,1,5,50,80$,
in \% of $n$) to estimate the likelihood. The quantities are the
mean $\mathrm{RED}_{1}$ and $\mathrm{RED}_{2}$ in \eqref{eq:RED}
measured with respect to computing time and number of evaluations,
respectively. Furthermore, $\bar{\sigma}_{R}$ is the mean (over MCMC
iterations) standard deviation of $\log(R_{m})$. Finally, $\alpha_{1}$
and $\alpha_{2}$ are the acceptance probabilities in \eqref{eq:alpha1}
and \eqref{eq:alpha2final} (expressed in \%), where the latter is
computed conditional on acceptance in the first stage. The corresponding
MH algorithm has an acceptance rate of $\approx$$23$\%.}

{\footnotesize{}}%
\begin{tabular}{llrrrrrrrrrr}
\toprule 
 &  & {\footnotesize{}$\mathrm{RED}_{1}$} &  & {\footnotesize{}$\mathrm{RED}_{2}$\textsubscript{}} &  & {\footnotesize{}$\bar{\sigma}_{R}$} &  & {\footnotesize{}$\alpha_{1}$} &  & {\footnotesize{}$\alpha_{2}$} & \tabularnewline
\cmidrule{3-3} \cmidrule{5-5} \cmidrule{7-7} \cmidrule{9-9} \cmidrule{11-11} 
 & {\footnotesize{}0.1} & {\footnotesize{}0.00} &  & {\footnotesize{}0.00} &  & {\footnotesize{}101.94} &  & {\footnotesize{}21} &  & {\footnotesize{}0} & \tabularnewline
 & {\footnotesize{}1} & {\footnotesize{}0.36} &  & {\footnotesize{}0.37} &  & {\footnotesize{}13.81} &  & {\footnotesize{}18} &  & {\footnotesize{}3} & \tabularnewline
 & {\footnotesize{}5} & {\footnotesize{}1.39} &  & {\footnotesize{}1.45} &  & {\footnotesize{}3.52} &  & {\footnotesize{}22} &  & {\footnotesize{}29} & \tabularnewline
 & {\footnotesize{}50} & {\footnotesize{}1.13} &  & {\footnotesize{}1.33} &  & {\footnotesize{}0.63} &  & {\footnotesize{}24} &  & {\footnotesize{}80} & \tabularnewline
 & {\footnotesize{}80} & {\footnotesize{}0.91} &  & {\footnotesize{}1.08} &  & {\footnotesize{}0.31} &  & {\footnotesize{}24} &  & {\footnotesize{}90} & \tabularnewline
\bottomrule
\end{tabular}\label{tab:PM}
\end{table}
\begin{figure}[h]
\includegraphics[width=0.8\columnwidth]{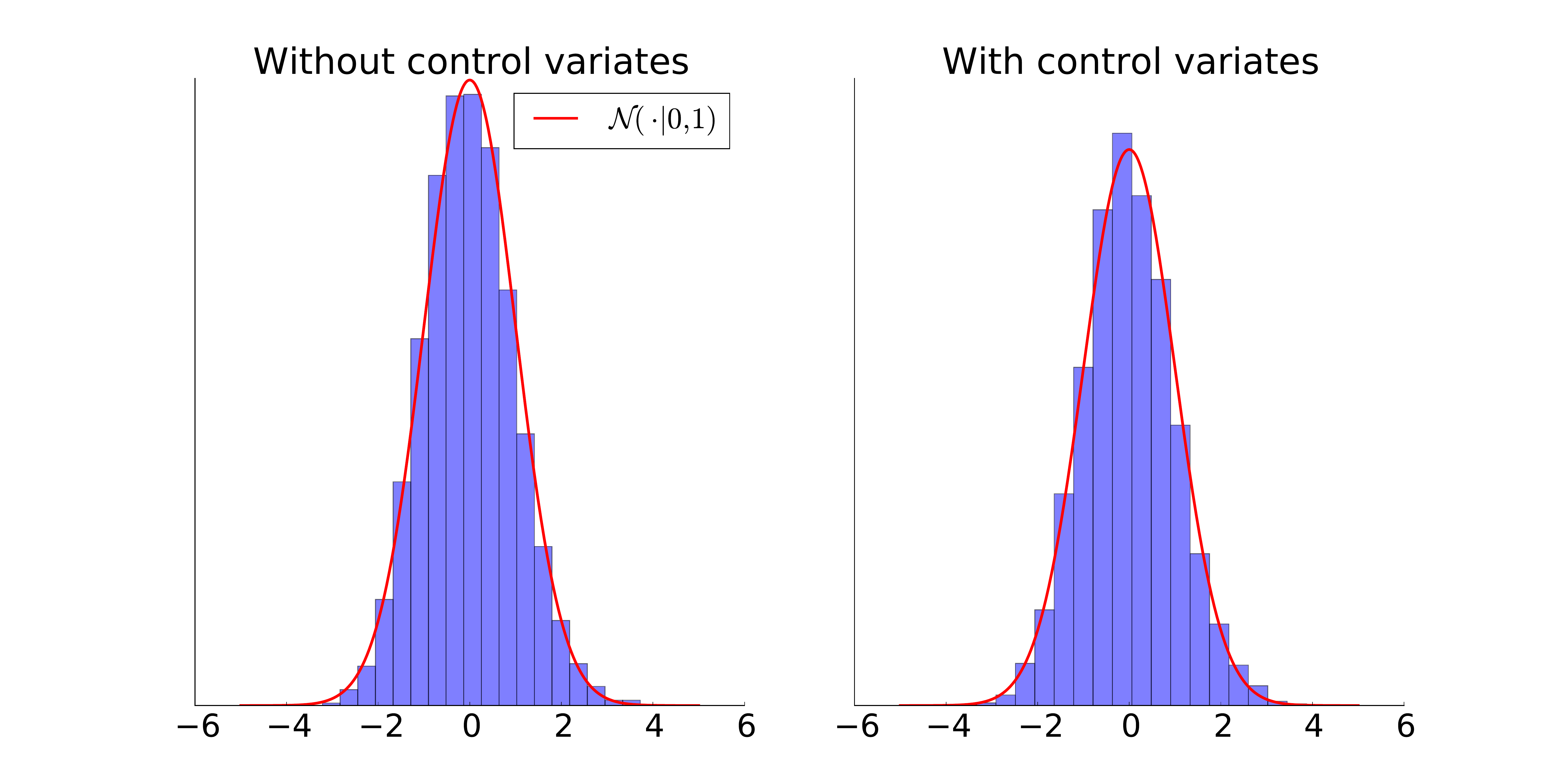}

\caption{\emph{Addressing the normality assumption in Theorem \ref{thm:theorem1}.
}The figure shows a histogram of (standardized) $\hat{l}_{m,n}(\theta_{c},\theta_{p})$
in \eqref{eq:diff_likelihood_est} ($10,000$ Monte Carlo replicates)
without control variates \citep[left,][]{payne2015bayesian} and with
control variates (right). The red solid line is the density function
of a standard normal variable. The validation is for the smallest
value of $m$ ($0.1$\% of $n$) for both cases and for the smallest
value of $K$ (less accurate, $0.03$\% of $n$) for the control variates.
The values of the parameters are $\theta_{c}=\theta^{\star}$, where
$\theta^{\star}$ is the mode, and $\theta_{p}$ is sampled from the
proposal distribution. We have verified the assumption for several
values of $\theta_{c}$ and $\theta_{p}$ (not shown here).}
\label{fig:NormalityExample1}
\end{figure}

\subsection{Example 2: DA-BPMMH}

Our second example explores how the state-dependent delayed acceptance
BPMMH improves the BPMMH proposed in \citet{quiroz2016speeding}.
The results are presented in Table \ref{tab:DelayedPMMH}. \citet{quiroz2016speeding}
in turn show that they outperform other subsampling approaches and,
in particular, that many of these approaches perform more poorly than
the standard MH (see also \citealp{bardenet2015markov}) in terms
of efficiency and/or can give a very poor approximation of the posterior.
Here we find that the delayed acceptance BPMMH is 30 times more efficient
than MH, which is a huge improvement considering the aforementioned
facts. Moreover, we find that the approximate posterior produced is
very close to the true posterior (simulated by MH), as illustrated
in Figure \ref{fig:KDEs}.

\begin{table}[h]
\centering \caption{\emph{Delayed acceptance block PMMH.} The table shows some quantities
for Block PMMH (BPMMH, \citealp{quiroz2016speeding}) and delayed
acceptance BPMMH. The latter is implemented with a Gaussian Process
(GP) and Linear Regression (LR) for learning $f(\theta)$ to predict
$e(\theta)$ in \eqref{eq:discrepency_between_stages}. The block
PMMH use a sample size of $0.5$\% which corresponds to $\sigma^{2}\approx10$.
The approximations used by BPMMH for computing the likelihood estimate
is based on $K=3.68$ (expressed as \% of $n$) number of clusters.
The delayed acceptance version uses an approximation based on $K^{(1)}=0.71$
(expressed as \% of $n$) in the first stage (and the same as BPMMH
for the second stage). The quantities are the mean $\mathrm{RED}_{2}$
in \eqref{eq:RED} measured with respect to number of evaluations.
Moreover, $\alpha_{1}$ and $\alpha_{2}$ are the acceptance probabilities
in \eqref{eq:alpha1} and \eqref{eq:alpha2final} (expressed in \%),
where the latter is computed conditional on acceptance in the first
stage. The corresponding BPMMH algorithm targets an acceptance rate
of $\approx$$10$\%, obtaining $8.4\%$ in this run. }

{\footnotesize{}}%
\begin{tabular}{llrrrrrrr}
\toprule 
 &  &  & {\footnotesize{}$\mathrm{RED}_{2}$\textsubscript{}} &  & {\footnotesize{}$\alpha_{1}$} &  & {\footnotesize{}$\alpha_{2}$} & \tabularnewline
\cmidrule{4-4} \cmidrule{6-6} \cmidrule{8-8} 
 & {\footnotesize{}BPMMH} &  & {\footnotesize{}14.03} &  &  &  &  & \tabularnewline
 & {\footnotesize{}DA-BPMMH (GP)} &  & {\footnotesize{}25.53} &  & {\footnotesize{}9.6} &  & {\footnotesize{}95} & \tabularnewline
 & {\footnotesize{}DA-BPMMH (LR)} &  & {\footnotesize{}30.19} &  & {\footnotesize{}9.2} &  & {\footnotesize{}99} & \tabularnewline
\bottomrule
\end{tabular}\label{tab:DelayedPMMH}
\end{table}

\begin{figure}[h]
\includegraphics[width=1\columnwidth]{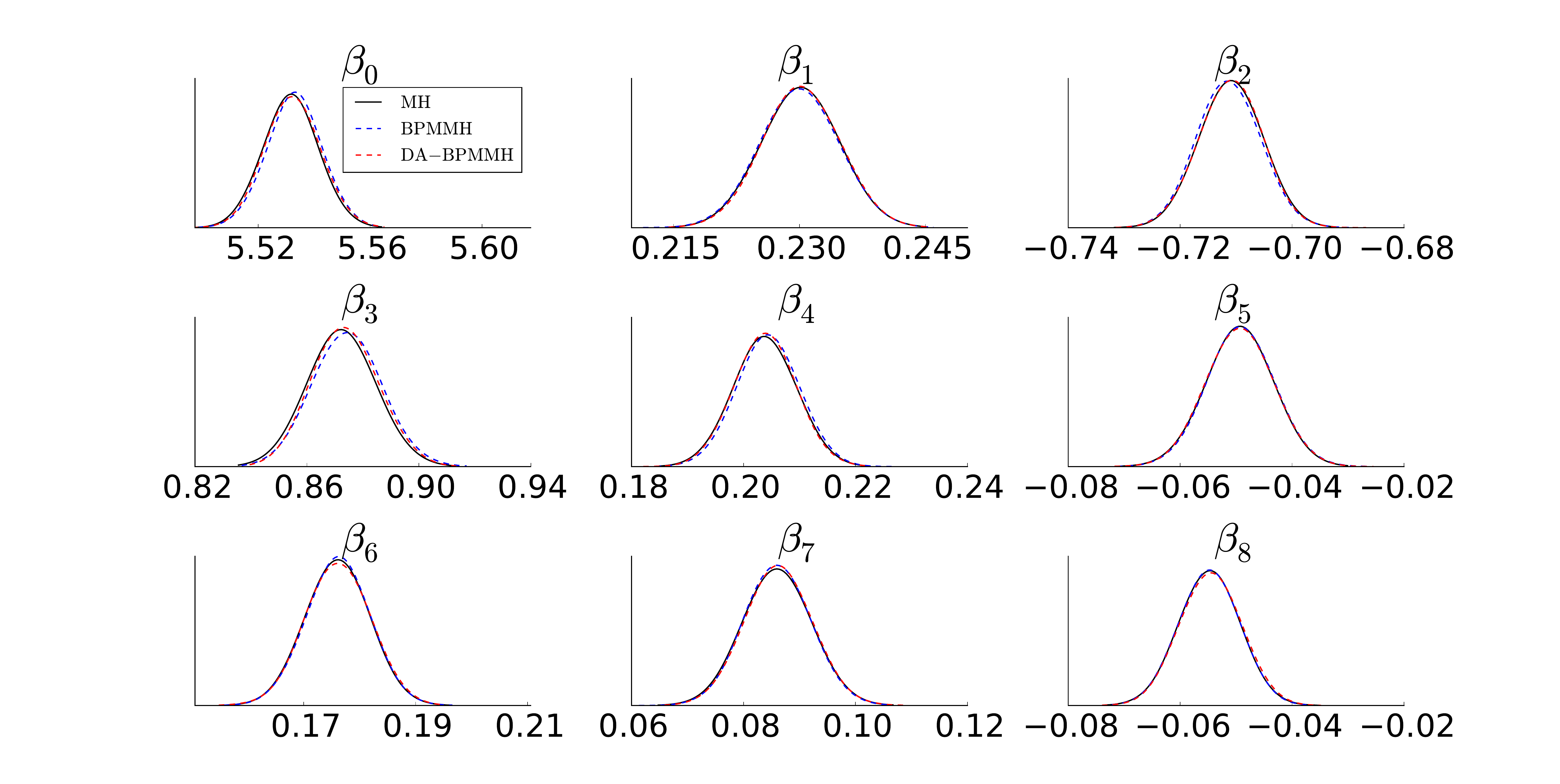}

\caption{\emph{Kernel density estimations of marginal posteriors}. The figure
shows the marginal posteriors with the different algorithms MH (standard
MH, solid black line), BPMMH (block PMMH, dashed blue line) and DA-BPMMH
(delayed acceptance BPMMH, dashed red line). The approximations used
by BPMMH for computing the likelihood estimate is based on $K=3.68$
(expressed as \% of $n$) number of clusters with $\sigma^{2}\approx10$.
The delayed acceptance version uses an approximation based on $K^{(1)}=0.71$
(expressed as \% of $n$) in the first stage.}
\label{fig:KDEs}
\end{figure}

\section{Conclusions\label{sec:Conclusions-and-Future}}

We explore the use of the efficient and robust difference estimator
in a delayed acceptance MH setting. The estimator incorporates auxiliary
information about the contribution to the log-likelihood function
while keeping the computational complexity low by operating on a sparse
set of the data. We demonstrate that the estimator is more efficient
than that of \citet{payne2015bayesian} in terms of having a much
lower variance. Moreover, we prove that a lower variability implies
that the delayed acceptance algorithm is more efficient, as measured
by the probability of accepting the second stage conditional that
the first stage is accepted. In an application to modeling of firm-bankruptcy,
we find that the proposed delayed acceptance algorithm is feasible
in the sense that it improves on the standard MH algorithm, which
is sometimes not true for \citet{payne2015bayesian}. We argue that
previous approaches (discussed in the introduction) of exact simulation
by MCMC are either (i) only possible under unfeasible assumptions
or (ii) very inefficient (compared to MH). We therefore believe that
exact simulation by subsampling MCMC is only possible by a delayed
acceptance approach, and the implementation provided here is crucial
for success.

Next, we realize that a delayed acceptance approach has the caveat
of scanning the complete data when deciding upon final acceptance.
We propose a state-dependent delayed acceptance that replaces the
second stage evaluation with an estimate. This algorithm inherently
allows for correlating the subsamples used for estimating the likelihood,
and we can leverage on recent advances in the pseudo-marginal literature
to reduce the computational cost. Moreover, we show that it is a special
case of the state-dependent delayed acceptance in \citet{christen2005mcmc}
and thus convergence to an invariant distribution follows. This distribution
is perturbed because the second stage estimate is biased, but we can
control the error and ensure that it is within $O(m^{-2})$ of the
true posterior. We demonstrate that the approximation is very accurate
and we can improve on MH by a factor of $30$ in terms of a measure
that balances statistical and computational efficiency.

\bibliographystyle{apalike}
\addcontentsline{toc}{section}{\refname}\bibliography{ref}

\appendix

\section{Proof of Theorem 1\label{sec:Comparing-the-Efficiency}}

\begin{proof}[Proof of Theorem 1]It follows from the normality assumption
that $X=R_{m}\sim\log\mathcal{N}(0,\sigma_{R}^{2})$ with density
\begin{eqnarray*}
f(x) & = & \frac{1}{x}\frac{1}{\sqrt{2\pi\sigma_{R}^{2}}}\exp\left(-\frac{1}{2\sigma_{R}^{2}}\log(x)^{2}\right),\quad x>0.
\end{eqnarray*}
 The expectation of the acceptance probability $\alpha_{2}(\theta_{c}\rightarrow\theta_{p})$
with respect to $X$ is
\begin{eqnarray*}
\mathrm{E}[\min(1,X)] & = & \int_{0}^{1}xf(x)dx+\int_{1}^{\infty}f(x)dx.
\end{eqnarray*}
Since $\text{median}(X)=1$ we obtain $\int_{1}^{\infty}f(x)dx=0.5$.
Now,
\begin{eqnarray*}
\int_{0}^{1}xf(x)dx & = & \int_{0}^{1}\frac{1}{\sqrt{2\pi\sigma_{R}^{2}}}\exp\left(-\frac{1}{2\sigma_{R}^{2}}\log(x)^{2}\right)dx\\
 & = & \exp\left(\sigma_{R}^{2}/2\right)\int_{-\infty}^{0}\frac{1}{\sqrt{2\pi\sigma_{R}^{2}}}\exp\left(-\frac{1}{2\sigma_{R}^{2}}(y-\sigma_{R}^{2})^{2}\right)dy,
\end{eqnarray*}
with $y=\log(x)$. The integrand is the pdf of $Y\sim\mathcal{N}(\sigma_{R}^{2},\sigma_{R}^{2})$
and thus
\begin{eqnarray*}
\mathrm{E}[\min(1,X)] & = & \exp\left(\sigma_{R}^{2}/2\right)\left(1-\Phi(\sigma_{R})\right)+0.5.
\end{eqnarray*}
We now show that $\mathrm{E}[\min(1,X)]$ is decreasing in $\sigma_{R}$.
We have that
\begin{eqnarray*}
\frac{d}{d\sigma_{R}}\mathrm{E}[\min(1,X)] & = & \exp\left(\sigma_{R}^{2}/2\right)\left(\sigma_{R}-\sigma_{R}\Phi(\sigma_{R})-\frac{1}{\sqrt{2\pi}}\right),
\end{eqnarray*}
and we can (numerically) compute the maximum of the expression in
brackets on the right which is $\approx-0.23$. Since $\exp\left(\sigma_{R}^{2}/2\right)>0$
it follows that $\frac{d}{d\sigma_{R}}\mathrm{E}[\min(1,X)]<0$ and
hence $\mathrm{E}[\alpha_{2}(\theta_{c}\rightarrow\theta_{p})]$ decreases
as a function of $\sigma_{R}$.

\end{proof}
\end{document}